\documentclass[ecta,nameyear]{econsocart}

\RequirePackage[colorlinks,citecolor=blue,linkcolor=blue,urlcolor=blue,pagebackref]{hyperref}
\RequirePackage{amsthm,amsmath,amsfonts,amssymb}

\usepackage[utf8]{inputenc}
\usepackage{mathrsfs}
\usepackage{mathtools}
\usepackage{cancel}
\usepackage{mathdots}
\usepackage{stmaryrd}
\usepackage{stackrel}
\usepackage{undertilde}
\usepackage{graphicx}
\PassOptionsToPackage{version=3}{mhchem}
\usepackage{mhchem}
\usepackage{esint}
\usepackage[dvipsnames]{xcolor}

\startlocaldefs
\theoremstyle{definition}

\newtheorem{lem}{Lemma}[section]
\newtheorem{thm}{Theorem}[section]

\newtheorem{prop}{Proposition}[section]

\newtheorem{assumption}{Assumption}[section]
\endlocaldefs

\begin{document}

\begin{frontmatter}
\title{Convergence of Computed Dynamic Models with Unbounded Shock\protect}
\runtitle{Convergence of Computed Dynamic Models with Unbounded Shock}
%\thankstext{T1}{Footnote to the title with the `thankstext' command.}

\begin{aug}
\author[id=au1,addressref={add1}]{\fnms{Kenichiro}~\snm{McAlinn}\ead[label=e1]{kenichiro.mcalinn@temple.edu}}
\author[id=au2,addressref={add2}]{\fnms{K\={o}saku}~\snm{Takanashi}\ead[label=e2]{kosaku.takanashi@riken.jp}}

\address[id=add1]{%
\orgdiv{Department of Statistical Science, Fox School of Business},
\orgname{Temple University}}

\address[id=add2]{%
\orgdiv{Center for Advanced Intelligence Project},
\orgname{Riken}}
\end{aug}

\begin{abstract}
This paper studies the asymptotic convergence of computed dynamic models when the shock is unbounded.
Most dynamic economic models lack a closed-form solution.
As such, approximate solutions by numerical methods are utilized.
Since the researcher cannot directly evaluate the exact policy function and the associated exact likelihood, it is imperative that the approximate likelihood  asymptotically converges-- as well as to know the conditions of convergence-- to the exact likelihood, in order to justify and validate its usage.
In this regard, \cite{Villaverde-Ramirez-Santos_06} show convergence of the likelihood, when the shock has compact support.
However, compact support implies that the shock is bounded, which is not an assumption met in most dynamic economic models, e.g., with normally distributed shocks.
This paper provides theoretical justification for most dynamic models used in the literature by showing the conditions for convergence of the approximate invariant measure obtained from numerical simulations to the exact invariant measure, thus providing the conditions for convergence of the likelihood.
\end{abstract}


\end{frontmatter}

\section{Introduction}
\textsc{This paper studies} the convergence of dynamic economic models.
While dynamic economic models have become a central tool for research and policy, most do not have a closed-form solution.
Due to this, the policy function of these economic models are approximated by numerical methods.
This approximation means that the researcher can only evaluate the approximated transition function associated with the approximated invariant measure, rather than the exact invariant measure implied by the exact transition function.
Given that the researcher cannot evaluate the exact measure, it is natural to ask whether the approximate measure converges to the exact measure, at least asymptotically.
If, for example, it does not converge, or the conditions for convergence are not met in practice, then the validity of the estimated economic model and its output comes into question.
It is, therefore, critical that there is a theoretical foundation that provides the conditions for convergence to justify the usage of these dynamic economic models.

As a response, much econometric analysis has been done to provide this theoretical foundation.
For example,  \cite{santos2004simulation} and \cite{Santos_05} provide the foundations of simulations of approximate solutions for stochastic dynamic models by studying its accuracy properties, showing that  the computed moments from the numerically approximated policy converge to the exact moments as the approximation errors of the computed solutions go to zero.
Further, and more relevant to this paper, \cite{Villaverde-Ramirez-Santos_06} extends the results of \cite{Santos_05} to the convergence of the likelihood of computed economic models, providing conditions for which the approximated likelihood functions converges to the exact likelihood.
While the convergence results in \cite{Villaverde-Ramirez-Santos_06} provide some justification for dynamic economic models, one assumption it employs to obtain their result is rarely met.
This assumption is the compactness of the state space, which implies that the support of the shock of a dynamical system is bounded.
Although this assumption is standard in the numerical literature, it excludes-- among others-- dynamical models with normally distributed shocks.
As assuming a normally distributed shock is  standard in empirical studies, in which the evaluation of the likelihood is done by the usage of the Kalman filter \citep{Smets-Wouters_07}, it is simply vital that the results in \cite{Villaverde-Ramirez-Santos_06} extend to non-compact support, i.e., unbounded shock.
For example, recent works by \cite{Stachurski_02}, \cite{Nishimura-Stachurski_05}, \cite{Kamihigashi_07} and \cite{Kamihigashi-Stachurski_16} study the asymptotic invariant measure of the stochastic neoclassical growth model without compactness of the shocks and states. 
The purpose of this paper is to relax the compactness assumption for the convergence of the approximated invariant measure, providing the theoretical foundation and justification for these models.

The rest of this paper is organized as follows. 
Section~\ref{sec:Model-Set-up-and} gives the set-up of dynamic economic models and preliminary of the Markov operator. 
Section~\ref{sec:Convergence-of-Invariant} presents our result on the convergence of the invariant measure.
In Section~\ref{sec:errorbounds}, we derive error bounds for these approximations.
Section~\ref{sec:convergence} presents our main result on the convergence of computed likelihoods.
%We conclude in Section~\ref{sec:conclusion}, with further discussions.

\section{Model Set-Up and Preliminaries\label{sec:Model-Set-up-and}}
We follow the set of notations and models in \citet{Santos_05}. The
equilibrium law of motion of the state variables is specified by a
dynamical system of the form 
\begin{align*}
s_{n+1} & =\varphi\left(s_{n},\varepsilon_{n+1}\right),\qquad n=0,1,2,\ldots,
\end{align*}
where $s_{n}$ is a vector of state variables that characterize the
evolution of the system. The vector $s_{n}$ belongs to a measurable
state space $\left(S,\mathcal{S}\right)$. We endow $S$ with its
relative Borel $\sigma$-algebra, which we denote by $\mathcal{S}$.
The variable $\varepsilon$ is an independent and identically distributed
shock, which is defined on the sample space $\left(E,\mathcal{E}\right)$.
The distribution of the shock $\varepsilon$ is given by a stochastic
kernel $Q:S\times E\rightarrow\left[0,1\right]$, where $Q\left(s,A\right)$
is the probability of realizing the event $A\in\mathcal{E}$, given
that the current state is $s\in S$.

Given a random dynamical system, one can define a transition probability
on the state space in the following way. Define the transition probability
function as 
\begin{align*}
P\left(s,A\right) & =Q\left(\left\{ \varepsilon\mid\varphi\left(s,\varepsilon\right)\in A\right\} \right).
\end{align*}
The transition function, $P:S\times\mathcal{S}\rightarrow\left[0,1\right]$,
is defined by 
\begin{alignat*}{1}
P\left(s,A\right) & =Q\left(s,\varphi^{-1}\left(A\right)\right).
\end{alignat*}
Let $B\left(S\right)$ be the set of all bounded $S$-measurable real
valued functions on $S$, with sup norm $\left|f\right|=\sup_{S}\left|f\left(s\right)\right|$.
The Markov operator associated with $P$ is defined as 
\begin{alignat}{1}
Tf\left(s\right) & \triangleq\int f\left(t\right)P\left(s,dt\right)\label{Def:M-operator}\\
 & =\int f\left(\varphi\left(s,\varepsilon\right)\right)Q\left(s,d\varepsilon\right)\nonumber 
\end{alignat}
For any given initial condition $\mu_{0}$ on $\mathcal{S}$, the
evolution of future probabilities, $\left\{ \mu_{n}\right\} $, can
be specified by the following operator $T^{*}$ that takes the space
\begin{align*}
\mu_{n+1} & =\left(T^{*}\mu_{n}\right)\left(A\right)={\displaystyle \int}P\left(s,A\right)\mu_{n}\left(ds\right),
\end{align*}
for all $A$ in $\mathcal{S}$ and $n\geq0$. The adjoint $T^{*}$of
$T$ is defined by the formula 
\begin{alignat*}{1}
\left(T^{*}\mu\right)\left(A\right) & =\int P\left(t,A\right)\mu\left(dt\right).
\end{alignat*}

We maintain the following basic assumptions.
\begin{assumption}
\label{A1}{\it The space of sets $S$ and $E$ are both locally compact
and $\sigma$-compact. }
\end{assumption}

Locally compact means that for each point $x\in S$, there is some
compact subspace $C$ of $S$ that contains a neighborhood of $x\in S$.
Further, $\sigma$-compact is a countable union of compact spaces.
Note that the space $\mathbb{R}^{d}$ is both locally compact and
$\sigma$-compact. A space that is both locally compact and $\sigma$-compact
can be written as an increasing union of countably many open sets,
each of which is compact and closed. In \citet{Santos_05}, they impose
the compactness assumption on both states, $S$ and $E$, which, again,
is not an assumption met in most dynamic economic models used in empirical
studies. In an important distinction, we relax this restriction to
the non-compact case, which allows us to use the whole Euclidean state,
$S=\mathbb{R}^{d}$, and unbounded distributions, such as the normal
distribution.

Recall that the probability measure $P$ is called \textit{tight}
if for all $\epsilon>0$ there is a compact set $K\subset\mathcal{S}$
such that $P\left(K\right)\geqq1-\epsilon$. Any probability measure
on the complete separable metric space is tight.
\begin{assumption}
\label{A2}{\it The Markov operator $T^{*}$ has a unique fixed point $\mu_{0}$:
$T^{*}\mu_{0}=\mu_{0}$. }
\end{assumption}

A sufficient condition for Assumption \ref{A2} is that there exists
a point, $s_{0}\in S$, such that, for any point $s\in S$, any neighborhood
$U$ of $s_{0}$ and any integer $k\geq1$, we have $P^{nk}\left(s,U\right)>0$
\citep[see][Section 3.2]{Futia_82}. 
\begin{assumption}
\label{A3}{\it Function $\varphi{:}\ S\times E\rightarrow S$ is jointly
measurable. Moreover, for every continuous function $f{:}\ S\rightarrow\mathbb{R}$,
\begin{align*}
{\displaystyle \int}f\left(\varphi\left(s_{j},\varepsilon\right)\right)Q\left(d\varepsilon\right) & \underset{}{\rightarrow}{\displaystyle \int}f\left(\varphi\left(s,\varepsilon\right)\right)Q\left(d\varepsilon\right)\;as\;s_{j}\rightarrow s.
\end{align*}}
\end{assumption}

Assumption \ref{A3} is the same as Assumption 2 in \citet{Santos_05}.

In most cases, the researcher does not know the exact form of the
transition equation $\varphi$, and only has access to the numerical
approximation of the transition equation, $\varphi_{j}$, with index
$j$. The index $j$ indicates the approximation and implies that,
as $j$ goes to infinity, the approximation, $\varphi_{j}$, converges
to the exact value (the metric of convergence is defined later). Every
numerical approximation $\varphi_{j}$ defines the transition probability
$P_{j}$ on $\left(S,\mathcal{S}\right)$. Given an approximation
$\varphi_{j}$, we define the corresponding approximation of the transition
probability as 
\begin{align*}
P_{j}\left(s,A\right) & =Q\left(\left\{ \varepsilon\mid\varphi_{j}\left(s,\varepsilon\right)\in A\right\} \right),
\end{align*}
and define the approximated transition function, $P_{j}:S\times\mathcal{S}\rightarrow\left[0,1\right]$,
as 
\begin{alignat*}{1}
P_{j}\left(s,A\right) & =Q\left(s,\varphi_{j}^{-1}\left(A\right)\right).
\end{alignat*}
The Markov operator associated with $P_{j}$ is defined as 
\begin{alignat}{1}
T_{j}f\left(s\right) & \triangleq\int f\left(t\right)P_{j}\left(s,dt\right)\label{Def:M-operator-1}\\
 & =\int f\left(\varphi_{j}\left(s,\varepsilon\right)\right)Q\left(s,d\varepsilon\right).\nonumber 
\end{alignat}
The evolution of future probabilities, $\left\{ \mu_{n}^{j}\right\} $,
can be specified by the following operator $T_{j}^{*}$ that takes
the space 
\begin{align*}
\mu_{n+1}^{j} & =\left(T_{j}^{*}\mu_{n}^{j}\right)\left(A\right)={\displaystyle \int}P_{j}\left(s,A\right)\mu_{n}^{j}\left(ds\right),
\end{align*}
for all $A$ in $\mathcal{S}$ and $n\geq0$. The adjoint $T_{j}^{*}$
of $T_{j}$ is defined by 
\begin{alignat*}{1}
\left(T_{j}^{*}\mu\right)\left(A\right) & =\int P_{j}\left(t,A\right)\mu\left(dt\right).
\end{alignat*}

Every numerical approximation $\varphi_{j}$ satisfies a structure
parallel to that of the above Assumptions \ref{A2} and \ref{A3}.
We further assume:
\begin{assumption}
\label{A2-1}{\it The Markov operator $T_{j}^{*}$ has a unique fixed point
$\mu_{0}^{j}$: $T_{j}^{*}\mu_{0}^{j}=\mu_{0}^{j}$, for all $j$.}
\end{assumption}

\begin{assumption}
\label{A3-1}{\it For each $j$, the function $\varphi_{j}{:}\ S\times E\rightarrow S$
is jointly measurable. Moreover, for every continuous function $f{:}\ S\rightarrow\mathbb{R}$,
\begin{align*}
{\displaystyle \int}f\left(\varphi_{j}\left(s_{j},\varepsilon\right)\right)Q\left(d\varepsilon\right) & \underset{}{\rightarrow}{\displaystyle \int}f\left(\varphi_{j}\left(s,\varepsilon\right)\right)Q\left(d\varepsilon\right)\;as\;s_{j}\rightarrow s.
\end{align*}}
\end{assumption}

\section{Convergence of the Invariant Distribution\label{sec:Convergence-of-Invariant}}
% 段落23から39までのソースコードをプレビューする

Now, recall the convergence of probability measures on $S$. When
the state space $S$ is separable, we can introduce a metric $D$
in the space of probability measures on $S$, such that $\lim_{n}D\left(\mu_{n},\mu\right)=0$
if and only if $\mu_{n}$ converges in law to $\mu$. Specifically,
the metric we use is the Fortet-Mourier metric \citep[][Section11.3]{Dudley_02}:
\begin{alignat}{1}
D\left(\mu_{n},\mu\right) & =\sup_{f\in BL\left(S\right)}\left|\int_{S}f\left(s\right)d\mu_{n}-\int_{S}f\left(s\right)d\mu\right|,\label{eq:FMmetric}
\end{alignat}
where the supremum $\sup_{f\in BL\left(S\right)}$ is taken over all
bounded Lipschitz continuous functions defined on $S$: $BL\left(S\right)$. 

The main question we answer in this paper is the following: How strong of a topology is sufficient for the approximate transition equation, $\varphi_{j}$, to converge to the true transition equation, $\varphi$, in order for the approximate invariant measure, $\mu_{n}$, to converge to the convergence in law distance eq.~(\ref{eq:FMmetric}).
\cite{Santos_05}, assuming that the state-space is compact, showed that convergence under the following topology is sufficient to prove the convergence of the invariant measure.
%さて，本論文での主題を改めて述べると次のようになる．近似された不変測度$\mu_{n}$が法則収束の距離(\ref{eq:FMmetric})で収束するためには，近似したthe
%transition equation, $\varphi_{j}$ を，どのくらいの強さの位相で，正しいthe transition
%equation, $\varphi$へ収束させれば十分か？\citet{Santos_05}は，状態空間にコンパクト性を仮定した上で，次の位相の下で収束させれば，不変測度の収束に十分であることを示した．they
Endow the metric in the space of functions $\varphi$ and $\hat{\varphi}$
as 
\begin{align*}
 & \max_{s\in S}\left[{\displaystyle \int}\left\Vert \varphi\left(s,\varepsilon\right)-\hat{\varphi}\left(s,\varepsilon\right)\right\Vert Q\left(d\varepsilon\right)\right]
\end{align*}
where $\left\Vert \cdot\right\Vert $ is the max norm in $\mathbb{R}^{l}$.
This metric only works under the compactness assumption on $S$.
To consider the functional approximation of the transition equation, $\varphi$, under non-compactness, this uniform topology is not practical.
In the following, we extend the state-space, $S$, to non-compactness and weaken the uniform convergence topology of the functional approximation to a local uniform topology.
 %また，非コンパクトの下でtransition
%equation $\varphi$を関数近似する上では，この一様位相は少々使いづらい．以下では，状態空間$S$を非コンパクトな場合へ拡張すると共に，関数近似の一様収束の位相をlocal
%な一様位相へ弱める．

First, note that $BL\left(S\right)$ can be relaxed to infinitely continuously
differentiable functions on $S$: $C^{\infty}\left(S\right)$ by using
the mollifier method. Then, we have the following lemma: 
\begin{lem}
\label{lem1} {\it The limit, $\lim_{n}D\left(\mu,\mu_{n}\right)=0$, holds
if and only if 
\begin{alignat*}{1}
\lim_{n\rightarrow\infty}\sup_{f\in C^{\infty}\left(S\right)}\left|\int_{S}f\left(s\right)d\mu_{n}-\int_{S}f\left(s\right)d\mu\right| & =0.
\end{alignat*}}
\end{lem}

Note that by Assumptions \ref{A2} and \ref{A3}, each $\varphi_{n}$
defines the associated pair $\left(P_{j},T_{j}\right)$. The adjoint
$T_{j}^{*}$ of $T_{j}$ is 
\begin{alignat*}{1}
\left\langle T_{j}f,\mu_{j}\right\rangle  & =\int\int f\left(\varphi_{j}\left(s,\varepsilon\right)\right)Q\left(s,d\varepsilon\right)d\mu_{j}\left(s\right)\\
\left\langle f,T_{j}^{*}\mu_{j}\right\rangle  & =\int\int f\left(\varphi_{j}\left(s,\varepsilon\right)\right)Q\left(s,d\varepsilon\right)d\mu_{j}\left(s\right).
\end{alignat*}
Moreover, there always exists an invariant distribution $\mu_{j}^{*}=T_{j}^{*}\mu_{j}^{*}$.

Given Lemma~\ref{lem1}, we have the following result: 
\begin{prop}
\label{prop1} {\it Suppose Assumptions \ref{A1}, \ref{A2}, and \ref{A3}
are satisfied for each approximated model $\varphi_{i}$ and $\varphi_{0}$.
Then, a sufficient condition for a sequence of the measure $\mu_{j}$,
associated with $T_{j}$, to converge to $\mu_{0}$, associated with
$T_{0}$, in the sense of eq.~(\ref{eq:FMmetric}), is the strong
convergence of $T_{j}$ to $T_{0}$. }
\end{prop}

Since $S$ is completely regular, it has Stone-Cech compactification
$\beta\left(S\right)$: 
\begin{thm}
\citep[Stone-Cech compactification:][Theorem 38.2]{Munkres_00}. {\it Let
$S$ be a completely regular space. Then, there exists a compactification
$\beta\left(S\right)$ of $S$ having the property that every bounded
continuous function $f{:}\ S\rightarrow\mathbb{R}$ extends uniquely
to a continuous function of $\beta\left(S\right)$ into $\mathbb{R}$.}
\end{thm}

The Stone-Cech compactification, $\beta\left(S\right)$, of which
$S$ is a dense subspace, satisfies the property that each bounded
continuous function, $f{:}\ S\rightarrow\mathbb{R}$, has a continuous
extension, $g:\beta\left(S\right)\rightarrow\mathbb{R}$. We endow
the metric in the space of functions defined on the locally compact
and $\sigma$-compact space, $S$. For any two vector-value functions
$\varphi$ and $\hat{\varphi}$, let $d\left(\cdot,\cdot\right)$
be 
\begin{align}
d\left(\varphi,\hat{\varphi}\right) & =\max_{f\in C^{\infty}\left(\beta\left(S\right)\right)}\max_{s\in\beta\left(S\right)}\left[{\displaystyle \int}\left|f\left(\varphi\left(s,\varepsilon\right)\right)-f\left(\hat{\varphi}\left(s,\varepsilon\right)\right)\right|Q\left(d\varepsilon\right)\right].\label{eq:metric}
\end{align}
The metric in eq.~(\ref{eq:metric}) is weaker than the metric  of \citet{Santos_05}, and extends to the non-compact state space. In this section,
convergence of the sequence of functions $\left\{ \varphi_{j}\right\} $
is in this distance, as this metric can accommodate the noncontinuous
functions $\varphi$ and $\hat{\varphi}$. Although we will impose
continuous differentiability on $\varphi$ for the convergence of
the approximate likelihood studied in Section~\ref{sec:convergence},
the metric $d\left(\cdot,\cdot\right)$ is sufficient to guarantee
the convergence of the invariant distribution. 

Then, we have the following theorem: 
\begin{thm}
\label{thm1} {\it Let $\left\{ \varphi_{j}\right\} $ be a sequence of
functions that converge to $\varphi$, in the sense of $d\left(\cdot,\cdot\right)$
in eq.~(\ref{eq:metric}). Let $\left\{ \mu_{j}^{*}\right\} $ be
a sequence of probabilities on $\mathcal{S}$ associated with $\left\{ \varphi_{j}\right\} $,
such that $\mu_{j}^{*}=T_{j}^{*}\mu_{j}^{*}$, for each $j$. Under
Assumptions \ref{A1} and \ref{A2}, if $\mu^{*}$ is a weak limit
point of $\left\{ \mu_{j}^{*}\right\} $, then $\mu^{*}=T^{*}\mu^{*}$. }
\end{thm}

This theorem asserts the bilinear convergence of $T_{j}^{*}\mu_{j}^{*}$
to $T^{*}\mu^{*}$ in the weak topology.

\section{Error Bounds\label{sec:errorbounds}}
% 段落41から48までのソースコードをプレビューする

In this section, we study the error bounds of these approximations
under non-compactness. The error bounds are important for two reasons.
First, in numerical applications, it is often desirable to bound the
size of the approximation error in order to know the theoretical limit
of the approximation. Second, computations cannot go on forever and
must stop in finite time. Hence, knowing the error bounds can dictate
an efficient stopping criteria to minimize computational cost while
ensuring convergence. As such, \citet{Santos_05} give a bound on
the size of the approximation error under the compactness assumption.

To begin, we introduce the notion of compactness for the Markov operator.
The Markov operator $T$ is \textit{compact} if the image $T\left(bX\right)$
has compact closure in $X$, where $bX=\left\{ x\in X|\left\Vert x\right\Vert \leq1\right\} $.
The Markov operator $T$ is \textit{quasi-compact} if there is a unique
compact operator $L$ and an integer $n$ such that 
\[
\sup_{x\in bX}\left\Vert T^{n}x-Lx\right\Vert <1.
\]
If the above quasi-compactness is satisfied, one can obtain the convergence
of the sequence of operators, $\left\{ T^{n}\right\} $, to the invariant
probability at a geometric rate. The following theorem gives this
result. 
\begin{thm}
\label{thm:Yosida-and-Kakutani}\citep{Yosida-Kakutani_41}. {\it Let $T$
be a stable, quasi-compact Markov operator defined by eq.~(\ref{Def:M-operator})
satisfying Assumptions \ref{A1}, \ref{A2}, and \ref{A3}. Then,
there exist constants, $C$, with $\varepsilon>0$, such that 
\begin{align*}
\sup_{s\in\beta\left(S\right)}\left\Vert T^{n}f\left(s\right)-T^{*}f\left(s\right)\right\Vert  & \leq\frac{C}{\left(1+\varepsilon\right)^{n}}.
\end{align*}}
\end{thm}

The following theorem bounds the approximation error between the expected
values of $f$ over the true invariant measure $\mu^{*}$ and the
approximate invariant measure $\hat{\mu}^{*}$ of $\hat{\varphi}$. 
\begin{prop}
\label{prop2} {\it Let $f$ be a Lipschitz function with constant $L$.
Suppose we have a numerical approximation $\hat{\varphi}$ with $d\left(\hat{\varphi},\varphi\right)\leq\delta$,
for some $\delta>0$. Then, 
\[
\left|\int f\left(s\right)\mu^{*}\left(ds\right)-\int f\left(s\right)\hat{\mu}^{*}\left(ds\right)\right|\leq\frac{Ld\left(\hat{\varphi},\varphi\right)}{\varepsilon},
\]
where $\mu^{*}$ is the unique invariant measure of the exact $\varphi$,
and $\hat{\mu}^{*}$ is a unique invariant measure of $\hat{\varphi}$. }
\end{prop}

Quasi-compact operators enjoy a very useful property in Theorem \ref{thm:Yosida-and-Kakutani}.
Furthermore, quasi-compact operators are easily recognizable. In fact,
we find that most operators are quasi-compact \citep[see][]{Futia_82}.
In our circumstance, Assumption~\ref{A2} guarantees the quasi-compactness
of the Markov operator.

% 段落49から80までのソースコードをプレビューする

\section{Application: Convergence of Computed Likelihood\label{sec:convergence}}

Given the convergence of the invariant measure in the previous section, we prove the convergence of the approximate likelihood in \citet{Villaverde-Ramirez-Santos_06}.
We relax the compactness assumption in the state-space and  shock, and prove equivalent results to \citet{Villaverde-Ramirez-Santos_06}, justifying the construction of the likelihood via the Kalman filter, among others. 
%前節で得られた不変測度の収束の結果をつかって，\citet{Villaverde-Ramirez-Santos_06}による近似尤度の収束を述べよう．状態空間およびショックのコンパクト性を外して\citet{Villaverde-Ramirez-Santos_06}と同様な結果を得ることができ，カルマンフィルタによる尤度の構成方法が正当化される．

\subsection{Likelihood Induced by Random Dynamical Systems}

The equilibrium law of motion of the state space system can be specified
as 
\begin{alignat}{1}
s_{t} & =\varphi(s_{t-1},\varepsilon_{t};\theta),\label{eq:VRS(1)-1}
\end{alignat}
\begin{alignat}{1}
y_{t} & =g(s_{t},\eta_{t};\theta),\label{eq:VRS(2)-1}
\end{alignat}
where eq.~(\ref{eq:VRS(1)-1}) is the transition equation, and eq.~(\ref{eq:VRS(2)-1})
is the measurement equation. Here, the variables $\varepsilon_{t}$
and $\eta_{t}$ are tight random elements and are independent and
identically distributed shocks with values in some Euclidean space,
with bounded and continuous densities. Their distribution is given
by the probability measure, $Q$, defined on a measurable space, $\left(E,\mathcal{E}\right)\subset\left(\mathbb{R}^{d},\mathcal{B}\left(\mathbb{R}^{d}\right)\right)$.
We do not impose the compactness on the support of $Q$ but impose
tightness, in order to deal with unbounded shocks, such as normally
distributed shocks. The parameter, $\theta\in\Theta\subset\mathbb{R}^{n}$,
is a vector of structural parameters and $\Theta$ is on a compact
set. The vector, $y_{t}$, is the observables in each period, $t$.
Let $\mathcal{Y}_{T}=\left\{ y_{t}\right\} _{t=1}^{T}$ with $Y^{0}=\left\{ \emptyset\right\} $.
To avoid singularity, we impose $\textrm{dim}(\varepsilon_{t})+\textrm{dim}(\eta_{t})\geq\textrm{dim}(Y_{t})$.
And we partition $\left\{ \varepsilon_{t}\right\} $ into $\varepsilon_{t}=(\varepsilon_{1,t},\varepsilon_{2,t})$,
such that $\textrm{dim}(\varepsilon_{2,t})+\textrm{dim}(\eta_{t})=\textrm{dim}(y_{t}).$
As in the previous section, we index the approximations by $j$, the
numerical approximation to the transition equations is $\varphi_{j}$,
and the measurement equations is $g_{j}$.

As with the previous section, we assume that each state-space system has an invariant measure and that invariance measure is absolute continuous with regard to a Lebesgue measure:
%前節と同様に各state space systemは不変測度を持ち，その不変測度はLebesgue測度に関して絶対連続とする：
\begin{assumption}
\label{A2-2}{\it For all $\theta$ and all $j$, there exists a unique
invariant distribution for $\mathcal{S}$, $\mu(S;\theta)$, and $\mu_{j}(S;\theta)$,
that has a Radon-Nikodym derivative with respect to the Lebesgue measure. }
\end{assumption}

The exact likelihood is constructed, using the change of variables formula, as follows.
First we assume that the system can solve the error term, exactly.
%Exactな尤度は，変数変換公式によって，以下のように構成される．まずシステムが誤差項に関してexactに解けると仮定する．
\begin{assumption}
\label{A3-2}{\it For all $\theta$ and $t$, the system of equations eq.~(\ref{eq:VRS(1)-1})
and eq.~(\ref{eq:VRS(2)-1}) have a unique solution, 
\begin{alignat*}{1}
\eta_{t} & =\eta^{t}(\varepsilon_{1}^{t},s_{0},y^{t};\theta),\\
s_{t} & =\mathit{s}^{t}(\varepsilon_{1}^{t},s_{0},y^{t};\theta),\\
\varepsilon_{2,t} & =\varepsilon_{2}^{t}(\varepsilon_{1}^{t},s_{0},y^{t};\theta),
\end{alignat*}
and we can evaluate $p(\mathit{v}^{t}(W_{1}^{t},S_{0},y^{t};\theta);\theta)$
and $p(\mathit{w_{2}}^{t}(W_{1}^{t},S_{0},y^{t};\theta);\theta)$
for all $S_{0}$, $W_{1}^{t},$ and $t$. }
\end{assumption}

We further assume that the observation equation \textit{(\ref{eq:VRS(2)-1})} is continuously differentiable.
%また，観測方程式\textit{(\ref{eq:VRS(2)-1})}に関して，連続微分可能とする．
\begin{assumption}
\label{A1-1}{ \it For all $\theta$, function $g(\cdot,\cdot,\theta)$
is continuously differentiable, with bounded partial derivatives. }
\end{assumption}

From Assumptions \ref{A1-1}, \ref{A2-2}, and \ref{A3-2}, we construct
the likelihood function by the change of variable formula: 
\begin{equation}
p(y_{t}\mid W_{1}^{t},S_{0},y^{t-1};\theta)=p(\mathit{v}_{t};\theta)p(\mathit{w}_{2,t};\theta)\mid dy(\mathit{v}_{t},\mathit{w}_{2,t};\theta)\mid,\label{eq:CoV-1}
\end{equation}
where 
\[
\mid dy(\mathit{v}_{t},\mathit{w}_{2,t};\theta)\mid=\textrm{det}\left[\begin{array}{cc}
\nabla\frac{\partial g}{\partial v_{t}} & \nabla\frac{\partial g}{\partial w_{2,t}}\end{array}\right].
\]
Further, we have the following assumption \citep[][Assumption 4]{Villaverde-Ramirez-Santos_06}. 
\begin{assumption}
\label{A5} {\it For all $\theta$ and $t$, the model gives some positive
probability to the data $y^{T}$, that is, $p(y_{t}\mid W_{1}^{t},S_{0},y^{t-1};\theta)>\xi\geq0$
for all $S_{0}$ and $W_{1}^{t}.$ }
\end{assumption}

From Assumption~\ref{A5}, the likelihood is as follows, 
\begin{alignat*}{1}
L(y^{T};\gamma) & ={\displaystyle \prod_{t=1}^{T}p(y_{t}\mid y^{t-1};\theta)}\\
 & ={\displaystyle \prod_{t=1}^{T}\int\int p(y_{t}\mid W_{1}^{t},S_{0},y^{t-1};\theta)p(W_{1}^{t},S_{0}\mid y^{t-1};\theta)dW_{1}^{t}dS_{0}}\\
 & =\int\left(\int\prod_{t=1}^{T}p(W_{1}^{t};\theta)p(y_{t}\mid W_{1}^{t},S_{0},y^{t-1};\theta)dW_{1}^{t}\right)\mu^{*}(dS_{0};\theta).
\end{alignat*}

Next, we also assume that, also for the approximate state-space functions, $\left\{ \varphi_{j}\right\} $, and measurement function, $\left\{ g_{j}\right\} $, the system can solve the error term, exactly.
%つぎに近似されたstate space functions $\left\{ \varphi_{j}\right\} $およびmeasurement
%function $\left\{ g_{j}\right\} $についても同様に，システムが誤差項に関してexactに解けると仮定する．
\begin{assumption}
\label{A4}{\it For all $j$, the system of equations 
\begin{alignat*}{1}
S_{1} & =\varphi_{j}(S_{0},(W_{1,1},W_{2,1});\theta),\\
y_{m} & =g_{j}(S_{m},V_{m};\theta)\quad\textrm{for}\quad m=1,2,\dots,t,\\
S_{m} & =\varphi_{j}(S_{m-1},(W_{1,m},W_{2,m});\theta)\quad\textrm{for}\quad m=2,3,\dots,t,
\end{alignat*}
has a unique solution, 
\begin{alignat*}{1}
V_{j,t} & =\mathit{v}_{j}^{t}(W_{1}^{t},S_{0},y^{t};\theta),\\
S_{j,t} & =\mathit{s}_{j}^{t}(W_{1}^{t},S_{0},y^{t};\theta),\\
W_{j,2,t} & =\mathit{w}_{j,2}^{t}(W_{1}^{t},S_{0},y^{t};\theta),
\end{alignat*}
and we can evaluate $p(\mathit{v}_{j}^{t}(W_{1}^{t},S_{0},y^{t};\theta);\theta)$
and $p(\mathit{w_{j,2}}^{t}(W_{1}^{t},S_{0},y^{t};\theta);\theta)$
for all $S_{0}$, $W_{1}^{t},$ and $t$. }
\end{assumption}

We also assume that the measurement function, $\left\{ g_{j}\right\} $, is continuously differentiable.
%measurement functions $\left\{ g_{j}\right\} $に関しても連続微分可能を仮定する．
\begin{assumption}
{\it For all $j$, functions $g_{j}(\cdot,\cdot,\theta)$ are continuously
differentiable at all points except at a finite number of points.
At the points of differentiability, all partial derivatives are bounded,
and the bounds are independent of $j$. }
\end{assumption}

Then, $dy_{j}(v_{j,t},w_{j,2};\theta)$ exists for all but a finite
set of $S_{0}$ and $W_{1}^{t}$, we have, for all $j$, $\theta$,
and $t$, 
\[
p_{j}(y_{t}\mid W_{1}^{t},S_{0},y^{t-1};\theta)=p(\mathit{v}_{j,t};\theta)p(\mathit{w}_{j,2,t};\theta)\mid dy_{j}(\mathit{v}_{j,t},\mathit{w}_{j,2,t};\theta)\mid,
\]
where 
\[
\mid dy_{j}(\mathit{v}_{j,t},\mathit{w}_{j,2,t};\theta)\mid=\textrm{det}\left[\begin{array}{cc}
\nabla\frac{\partial g_{j}}{\partial v_{j,t}} & \nabla\frac{\partial g_{j}}{\partial w_{j,2,t}}\end{array}\right],
\]
for all $S_{0}$ and $W_{1}^{t}$, but a finite number of points.

As $j$ goes to infinity, $\varphi_{j}$ and $g_{j}$ converge to
their exact values. Unlike the previous section, convergence of the
sequence of functions $\left\{ \varphi_{j}\right\} $ and $\left\{ g_{j}\right\} $
need to be a stronger topology, which is defined in the following
way. For any two vector-valued functions $\varphi$ and $\hat{\varphi}$,
let 
\begin{alignat*}{1}
d_{C^{1}}\left(\varphi,\hat{\varphi}\right) & =\max_{i\in I}\sup_{s\in S_{i}}\left[\int\left\Vert \varphi\left(s,\varepsilon\right)-\hat{\varphi}\left(s,\varepsilon\right)\right\Vert Q\left(d\varepsilon\right)+\int\left\Vert \nabla\varphi\left(s\right)-\nabla\hat{\varphi}\left(s\right)\right\Vert Q\left(d\varepsilon\right)\right],
\end{alignat*}
where $\left\{ S_{i}\right\} _{i\in I}$ is an exhaustive sequence
of compact sets of $S$. In this section, convergence of a sequence
of functions, $\left\{ \varphi_{j},g_{j}\right\} $, should be understood
in this norm.

Assumption~\ref{A4} is required for the change of variable formula
in eq.~(\ref{eq:CoV-1}). Convergence in $C^{1}$ implies the convergence
of the solutions, $v_{j}^{t}$, $s_{j}^{t},$ and $w_{j,2}^{t}$,
from the same argument in the proof of consistency of Z-estimators.
More importantly, the convergence of the Jacobian, $\left|dy_{j}\left(\mathit{v}_{j,t},\mathit{w}_{j,2,t};\theta\right)\right|$,
can be derived from the Ascoli-Arzela theorem.

\begin{thm} (Ascoli-Arzela). \textit{Let $S$ be a separable $\sigma$-compact
metric space and $S_{i}$, $i\in I$, be an exhausting sequence of
compact sets. Further, let $C^{1}(S)$ be the Banach space of complex-valued
continuous functions, $f(s)$, normed by 
\begin{alignat*}{1}
\parallel f\parallel & =\max_{i\in I}\left[\sup_{s\in S_{i}}\mid f(s)\mid+\sup_{s\in S_{i}}\left|f^{\prime}\left(s\right)\right|\right].
\end{alignat*}
If the following two conditions are satisfied: 
\[
\begin{array}{cc}
{\displaystyle \sup_{n\geq1}\sup_{s\in S}\mid f_{n}(s)\mid<\infty}, & \textrm{(equi-bounded)}\\
{\displaystyle \lim_{\delta\downarrow0}\sup_{n\geq1,dis(s^{\prime},s^{\prime\prime})\leq\delta}\mid f_{n}(s^{\prime})-f_{n}(s^{\prime\prime})\mid=0}, & \textrm{(equi-continuous)}
\end{array}
\]
then the sequence, $\left\{ f_{n}(s)\right\} \subseteq C(S)$, is
relatively compact in $C(S)$.} \end{thm}

Given this, we have the following proposition that proves the convergence
of the approximate likelihood to the exact likelihood with unbounded
shock. 

\begin{prop}\label{prop3} \textit{Suppose Assumptions~\ref{A1-1}-\ref{A4},
and suppose $d_{C^{1}}\left(\varphi_{j},\varphi\right)\rightarrow0$
and $d_{C^{1}}\left(g_{j},g\right)\rightarrow0$ as $j\rightarrow\infty$.
Then, it holds that 
\[
{\displaystyle \prod_{t=1}^{T}p_{j}(y_{t}\mid y^{t-1};\gamma)\overset{\mathit{l}_{\infty}}{\longrightarrow}\prod_{t=1}^{T}p(y_{t}\mid y^{t-1};\gamma)}.
\]
} \end{prop}

%Step1- Convergence of $p_{j}\rightarrow p\Longleftrightarrow p(\mathit{v}_{j,t};\gamma)p(\mathit{w}_{j,2,t};\gamma)\mid dy_{j}(\mathit{v}_{j,t},\mathit{w}_{j,2,t};\gamma)\mid\rightarrow p(\mathit{v}_{t};\gamma)p(\mathit{w}_{2,t};\gamma)\mid dy(\mathit{v}_{t},\mathit{w}_{2,t};\gamma)\mid$
%\[
%\mid dy(\mathit{v}_{t},\mathit{w}_{2,t};\gamma)\mid=det\left[\begin{array}{cc}
%\nabla\frac{\partial g}{\partial v_{t}} & \nabla\frac{\partial g}{\partial w_{2,t}}\end{array}\right].
%\]

The result shows that, as the researcher gets better approximations
of the policy function in a dynamic economic model, the computed likelihood
converges to the exact likelihood, even if the shock is unbounded.
This result goes beyond the result in \citet{Villaverde-Ramirez-Santos_06}
and is particularly relevant to researchers using dynamic economic
models with unbounded shocks, such as normally distributed shocks--
a standard specification in the literature-- as it guarantees, asymptotically,
that the likelihood function implied by the model is the correct object
of interest.

%\section{Conclusion\label{sec:conclusion}}
\bibliographystyle{econometrica} % Style BST file (imsart-number.bst or imsart-nameyear.bst)
\bibliography{reference}       % Bibliography file (usually '*.bib')

\appendix

\section*{Appendix}\label{app}

\begin{proof}[Proof of Lemma~\ref{lem1}] The \textquotedblleft Only
    if\textquotedblright{} part is derived from $C^{\infty}\subset BL$.
    In this proof we will show the \textquotedblleft if\textquotedblright{}
    part. For any $f\in BL\left(S\right)$ and $\varepsilon>0$, there
    exists $u\in C^{\infty}\left(S\right)$, such that $\sup_{s\in S}\left|f\left(s\right)-u\left(s\right)\right|<\varepsilon$
    by the mollifier method. By triangular inequality, we have 
    \begin{alignat*}{1}
     & \sup_{f\in BL\left(S\right)}\left|\int_{S}f\left(s\right)d\mu_{n}-\int_{S}f\left(s\right)d\mu\right|\\
    \leq & \sup_{f\in BL\left(S\right)}\left|\int_{S}f\left(s\right)d\mu_{n}-\int_{S}u\left(s\right)d\mu_{n}\right|+\sup_{u\in C^{\infty}\left(S\right)}\left|\int_{S}u\left(s\right)d\mu_{n}-\int_{S}u\left(s\right)d\mu\right|\\
     & +\sup_{f\in BL\left(S\right)}\left|\int_{S}u\left(s\right)d\mu-\int_{S}f\left(s\right)d\mu\right|.
    \end{alignat*}
    The first and third term of the right hand side of the inequality
    are smaller than $\varepsilon$. The second term is given in the definition.
\end{proof}
    
\begin{proof}[Proof of Proposition~\ref{prop1}] The strong convergence
    of the sequence of operators is 
    \begin{alignat*}{1}
    \sup_{s\in S}\left\Vert Tf\left(s\right)-T_{j}f\left(s\right)\right\Vert  & =\sup_{s\in S}\left|\int f\left(\varphi\left(s,\varepsilon\right)\right)Q\left(s,d\varepsilon\right)-\int f\left(\varphi_{j}\left(s,\varepsilon\right)\right)Q\left(s,d\varepsilon\right)\right|\\
     & =\sup_{s\in S}\left|\int\left\{ f\left(\varphi\left(s,\varepsilon\right)\right)-f\left(\varphi_{j}\left(s,\varepsilon\right)\right)\right\} Q\left(s,d\varepsilon\right)\right|\\
     & \rightarrow0
    \end{alignat*}
    for all $f\in C^{2}\left(S\right)$. This simply means that 
    \begin{alignat*}{2}
    \lim_{j\rightarrow\infty}\mathbb{E}\left[f\left(\varphi_{j}\left(s\right)\right)\right] & =\mathbb{E}\left[f\left(\varphi\left(s\right)\right)\right], & \ ^{\forall}f\in C^{2}\left(S\right).
    \end{alignat*}
    Let $f$ belong to $\mathcal{A}$. Then, for any two $T$ and $T_{j}$,
    and corresponding invariant measures $\mu$ and $\mu_{j}$, we have
    \begin{alignat}{1}
    \left|\left\langle Tf,\hat{\mu}\right\rangle -\left\langle T_{j}f,\mu_{j}\right\rangle \right| & \leq\left|\left\langle Tf,\mu\right\rangle -\left\langle Tf,\mu_{j}\right\rangle \right|+\left|\left\langle Tf,\mu_{j}\right\rangle -\left\langle T_{j}f,\mu_{j}\right\rangle \right|.\label{eq:adjoint}
    \end{alignat}
    It follows from Proholov's Theorem that $\left\{ \mu_{j}\right\} $
    has a weakly convergent subsequence. Let $\left\{ \mu_{j_{k}}\right\} $
    be such a subsequence, and let $\hat{\mu}$ be its limit. Then for
    the first term in eq.~(\ref{eq:adjoint}), we have 
    \begin{alignat*}{1}
    \left|\left\langle Tf,\hat{\mu}\right\rangle -\left\langle Tf,\mu_{j_{k}}\right\rangle \right| & \rightarrow0.
    \end{alignat*}
    Next, we show the equality: $\hat{\mu}=\mu_{0}$. We have 
    \begin{alignat*}{2}
    \left|\left\langle f,\hat{\mu}\right\rangle -\left\langle T_{0}f,\hat{\mu}\right\rangle \right| & \leq & \left|\left\langle f,\hat{\mu}\right\rangle -\left\langle f,\mu_{j_{k}}\right\rangle \right|+\left|\left\langle f,\mu_{j_{k}}\right\rangle -\left\langle T_{0}f,\mu_{j_{k}}\right\rangle \right|\\
     &  & +\left|\left\langle T_{0}f,\mu_{j_{k}}\right\rangle -\left\langle T_{0}f,\hat{\mu}\right\rangle \right|.
    \end{alignat*}
    Since $Tf$ and $T_{0}f$ are continuous and $\left\{ \mu_{j_{k}}\right\} $
    convergence weakly to $\hat{\mu}$, the first and third terms on the
    right hand side approaches zero as $j$ goes to infinity. For the
    second term, we have 
    \begin{alignat*}{1}
    \left|\left\langle f,\mu_{j_{k}}\right\rangle -\left\langle T_{0}f,\mu_{j_{k}}\right\rangle \right| & =\left|\left\langle f,T_{j_{k}}^{*}\mu_{j_{k}}\right\rangle -\left\langle T_{0}f,\mu_{j_{k}}\right\rangle \right|\\
     & =\left|\left\langle T_{j_{k}}f,\mu_{j_{k}}\right\rangle -\left\langle T_{0}f,\mu_{j_{k}}\right\rangle \right|\\
     & \leq\left\Vert T_{j_{k}}f-T_{0}f\right\Vert 
    \end{alignat*}
\end{proof}
    
\begin{proof}[Proof of Theorem~\ref{thm1}] The topology of weak
    convergence can be defined by the metric on the probability measure
    space, 
    \begin{alignat*}{1}
    d\left(\mu,\nu\right) & =\sup_{f\in\mathcal{A}}\left\{ \left|\int f\left(s\right)\mu\left(ds\right)-\int f\left(s\right)\nu\left(ds\right)\right|\right\} ,
    \end{alignat*}
    where $\mathcal{A}$ is the space of Lipschitz functions on $S$,
    with constant $L\leq1$ and $-1\leq f\leq1$. Then, we have 
    \begin{alignat*}{1}
    \left\Vert Tf\left(s\right)-T_{j}f\left(s\right)\right\Vert  & =\left|\int f\left(\varphi\left(s,\varepsilon\right)\right)Q\left(d\varepsilon\right)-\int f\left(\varphi_{j}\left(s,\varepsilon\right)\right)Q\left(d\varepsilon\right)\right|\\
     & =\left|\int f\left(\varphi\left(s,\varepsilon\right)\right)-f\left(\varphi_{j}\left(s,\varepsilon\right)\right)Q\left(d\varepsilon\right)\right|.
    \end{alignat*}
    Since $f\in C^{\infty}$, there exists a constant $K$, such that
    \begin{alignat*}{1}
    \left|\int\left[f\left(\varphi\left(s,\varepsilon\right)\right)-f\left(\varphi_{j}\left(s,\varepsilon\right)\right)\right]Q\left(d\varepsilon\right)\right| & \leq Kd\left(\varphi,\varphi_{j}\right).
    \end{alignat*}
\end{proof}
    
\begin{proof}[Proof of Proposition~\ref{prop2}] Denote $s_{n}\left(s_{0}\right)$
    as 
    \[
    \underbrace{\varphi\left(\varphi\left(\varphi\cdots\left(\varphi\left(s_{0},\varepsilon_{1}\right),\varepsilon_{2}\right)\right)\right)}_{\textrm{n times}}.
    \]
    And denote $\hat{s}_{n}\left(s_{0}\right)$ as 
    \[
    \underbrace{\hat{\varphi}\left(\hat{\varphi}\left(\hat{\varphi}\cdots\left(\hat{\varphi}\left(s_{0},\varepsilon_{1}\right),\varepsilon_{2}\right)\right)\right)}_{\textrm{n times}}.
    \]
    Since $f$ is a Lipschitz function with constant $L$, we have 
    \begin{align*}
     & \left|\mathbb{E}\left[f\left(s_{n}\left(s_{0}\right)\right)\right]-\mathbb{E}\left[f\left(\hat{s}_{n}\left(s_{0}\right)\right)\right]\right|\\
    = & \left|\mathbb{E}\left[f\left(\varphi\left(s_{n-1}\left(s_{0}\right),\varepsilon_{n}\right)\right)\right]-\mathbb{E}\left[f\left(\hat{\varphi}\left(\hat{s}_{n-1}\left(s_{0}\right),\varepsilon_{n}\right)\right)\right]\right|\\
    \leq & \left|\mathbb{E}\left[f\left(\varphi\left(s_{n-1}\left(s_{0}\right),\varepsilon_{n}\right)\right)\right]-\mathbb{E}\left[f\left(\varphi\left(\hat{s}_{n-1}\left(s_{0}\right),\varepsilon_{n}\right)\right)\right]\right|\\
     & +\left|\mathbb{E}\left[f\left(\varphi\left(s_{n-1}\left(s_{0}\right),\varepsilon_{n}\right)\right)\right]-\mathbb{E}\left[f\left(\hat{\varphi}\left(\hat{s}_{n-1}\left(s_{0}\right),\varepsilon_{n}\right)\right)\right]\right|\\
    \leq & L\left(1-\varepsilon\right)\mathbb{E}\left\Vert s_{n-1}\left(s_{0}\right)-\hat{s}_{n-1}\left(s_{0}\right)\right\Vert +Ld\left(\hat{\varphi},\varphi\right).
    \end{align*}
    By the same argument above, we have 
    \begin{alignat*}{1}
     & L\left(1-\varepsilon\right)\mathbb{E}\left\Vert s_{n-1}\left(s_{0}\right)-\hat{s}_{n-1}\left(s_{0}\right)\right\Vert \\
    \leq & L\left(1-\varepsilon\right)^{2}\mathbb{E}\left\Vert s_{n-2}\left(s_{0}\right)-\hat{s}_{n-2}\left(s_{0}\right)\right\Vert +L\left(1-\varepsilon\right)^{2}d\left(\hat{\varphi},\varphi\right).
    \end{alignat*}
    Iterating this, we obtain 
    \[
    \left|\mathbb{E}\left[f\left(s_{n}\left(s_{0}\right)\right)\right]-\mathbb{E}\left[f\left(\hat{s}_{n}\left(s_{0}\right)\right)\right]\right|\leq\frac{Ld\left(\hat{\varphi},\varphi\right)}{\varepsilon}.
    \]
    Integrating this by an invariant measure $\hat{\mu}^{*}$ of $\hat{\varphi}$,
    we have 
    \[
    \left|\int\mathbb{E}\left[f\left(s_{n}\left(s_{0}\right)\right)\right]\hat{\mu}^{*}\left(ds_{0}\right)-\int\mathbb{E}\left[f\left(\hat{s}_{n}\left(s_{0}\right)\right)\right]\hat{\mu}^{*}\left(ds_{0}\right)\right|\leq\frac{Ld\left(\hat{\varphi},\varphi\right)}{\varepsilon}.
    \]
    Note that the second term on the left-hand side is equal to $\int f\left(s\right)\hat{\mu}^{*}\left(ds\right)$.
    From Theorem~\ref{thm:Yosida-and-Kakutani}, for every $s_{0}$,
    $\mathbb{E}\left[f\left(s_{n}\left(s_{0}\right)\right)\right]$ converges
    uniformly on $\beta\left(S\right)$ to $\int f\left(s\right)\mu^{*}\left(ds\right)$.
    Finally, we obtain 
    \[
    \left|\int f\left(s\right)\mu^{*}\left(ds\right)-\int f\left(s\right)\hat{\mu}^{*}\left(ds\right)\right|\leq\frac{Ld\left(\hat{\varphi},\varphi\right)}{\varepsilon}.
    \]
\end{proof}
    
\begin{thm}\label{thm:Change of Variable}(Change of variable formula).
    \textit{Let $g{:}\ \mathbb{R}^{n}\rightarrow\mathbb{R}^{m}$ be Lipschitz
    continuous, $n\leq m$. Then, for each measurable function $p{:}\ \mathbb{R}^{n}\rightarrow\mathbb{R}$,
    \begin{alignat*}{1}
    \int_{\mathbb{R}^{n}}p\left(x\right)Jg\left(x\right)m\left(dx\right) & =\int_{\mathbb{R}^{n}}\left[\sum_{x\in g^{-1}\left(y\right)}p\left(x\right)\right]m\left(dy\right).
    \end{alignat*}
    } 
\end{thm}
    
\begin{proof}[Proof of Proposition~\ref{prop3}] Let 
    \begin{alignat*}{1}
    {\displaystyle \prod_{t=1}^{T}p_{j}(y_{t}\mid y^{t-1};\theta)} & =\int\left(\int\prod_{t=1}^{T}p(\varepsilon_{1}^{t};\theta)p_{j}(y_{t}\mid\varepsilon_{1}^{t},s_{0},y^{t-1};\theta)d\varepsilon_{1}^{t}\right)\mu_{j}^{*}(ds_{0};\theta),
    \end{alignat*}
    \begin{alignat*}{1}
    {\displaystyle \prod_{t=1}^{T}\tilde{p}_{j}(y_{t}\mid y^{t-1};\theta)} & =\int\left(\int\prod_{t=1}^{T}p(\varepsilon_{1}^{t};\theta)p(y_{t}\mid\varepsilon_{1}^{t},S_{0},y^{t-1};\theta)d\varepsilon_{1}^{t}\right)\mu_{j}^{*}(ds_{0};\theta),
    \end{alignat*}
    \begin{alignat*}{1}
    {\displaystyle \prod_{t=1}^{T}p(y_{t}\mid y^{t-1};\theta)} & =\int\left(\int\prod_{t=1}^{T}p(\varepsilon_{1}^{t};\theta)p(y_{t}\mid\varepsilon_{1}^{t},s_{0},y^{t-1};\theta)d\varepsilon_{1}^{t}\right)\mu^{*}(ds_{0};\theta).
    \end{alignat*}
    Consider the convergence of 
    \begin{alignat*}{1}
    p_{j}\left(y_{t}\mid\varepsilon_{1}^{t},s_{0},y^{t-1};\theta\right) & =p(\eta_{j,t};\theta)p(\varepsilon_{j,2,t};\theta)\left|dy_{j}(\eta_{j,t},\varepsilon_{j,2,t};\theta)\right|,
    \end{alignat*}
    where 
    \[
    \mid dy_{j}(\mathit{v}_{t},\mathit{w}_{2,t};\gamma)\mid=det\left[\begin{array}{cc}
    \nabla\frac{\partial g_{j}}{\partial v_{t}} & \nabla\frac{\partial g_{j}}{\partial w_{2,t}}\end{array}\right].
    \]
    Since the change of variable formula eq.~(\ref{thm:Change of Variable})
    holds for $Jg_{j}\left(\varepsilon,\eta;\theta\right)=\left|dy_{j}(\mathit{v}_{j,t},\mathit{w}_{j,2,t};\theta)\right|$,
    we have 
    \[
    \int p_{j}\left(y_{t}\mid W_{1}^{t},S_{0},y^{t-1};\theta\right)Q\left(d\varepsilon,d\eta\right)=\int p\left(\varepsilon;\theta\right)p\left(\eta;\theta\right)Jg_{j}\left(\varepsilon,\eta;\theta\right)Q\left(d\varepsilon,d\eta\right).
    \]
    By $d_{C^{1}}\left(\varphi_{j},\varphi\right)\rightarrow0$ and $d_{C^{1}}\left(g_{j},g\right)\rightarrow0$,
    we have 
    \[
    \int p\left(\varepsilon;\theta\right)p\left(\eta;\theta\right)Jg_{j}\left(\varepsilon,\eta;\theta\right)Q\left(d\varepsilon,d\eta\right)\rightarrow\int p\left(\varepsilon;\theta\right)p\left(\eta;\theta\right)Jg\left(\varepsilon,\eta;\theta\right)Q\left(d\varepsilon,d\eta\right).
    \]
    Then, we have 
    \[
    \prod_{t=1}^{T}p_{j}(y_{t}\mid y^{t-1};\theta)\rightarrow{\displaystyle \prod_{t=1}^{T}\tilde{p}_{j}(y_{t}\mid y^{t-1};\theta)}.
    \]
    
    By Lemma~1 of \citet{Villaverde-Ramirez-Santos_06}, $\left(\int\prod_{t=1}^{T}p(W_{1}^{t};\theta)p(y_{t}\mid W_{1}^{t},S_{0},y^{t-1};\theta)dW_{1}^{t}\right)$
    is a continuous function. Further, given Proposition~1 of \citet{Villaverde-Ramirez-Santos_06},
    we have $\mu_{j}^{*}(dS_{0};\theta)\rightarrow\mu^{*}(dS_{0};\theta)$.
    Therefore, 
    \begin{alignat*}{1}
    {\displaystyle \prod_{t=1}^{T}\tilde{p}_{j}(y_{t}\mid y^{t-1};\theta)} & \rightarrow{\displaystyle \prod_{t=1}^{T}p(y_{t}\mid y^{t-1};\theta)}.
    \end{alignat*}
\end{proof}

%\section*{Acknowledgements}

\end{document}